\documentclass[11pt]{article}
\usepackage{fullpage}
\usepackage{amsfonts, amsthm, amssymb, amsmath}
\usepackage[pdftex, pagebackref=true, colorlinks=true, urlcolor=blue, citecolor=blue, linkcolor=blue]{hyperref}

\usepackage{framed}
\usepackage{verbatim}
\usepackage{fullpage}
\usepackage{epsfig}

\newtheorem{thm}{Theorem}
\newtheorem{lem}[thm]{Lemma}
\newtheorem{cor}[thm]{Corollary}
\newtheorem{defn}[thm]{Definition}
\newtheorem{clm}[thm]{Claim}
\newtheorem{cons}[thm]{Construction}
\newtheorem{prop}[thm]{Proposition}
\newtheorem{rem}[thm]{Remark}
\newtheorem{sub}[thm]{Subroutine}

\newenvironment{theorem}{\begin{thm}\begin{rm}}%
{\end{rm}\end{thm}}
{\end{rm}\end{lem}}
{\end{rm}\end{cor}}
\newenvironment{definition}{\begin{defn}\begin{em}}%
{\end{em}\end{defn}}
{\end{rm}\end{clm}}
{\end{em}\end{cons}}
{\end{em}\end{prop}}

\newcommand{\secref}[1]{\hyperref[#1]{Section \ref{#1}}}
\newcommand{\apref}[1]{\hyperref[#1]{\ref{#1}}}
\newcommand{\thref}[1]{\hyperref[#1]{Theorem \ref{#1}}}
\newcommand{\defref}[1]{\hyperref[#1]{Definition \ref{#1}}}
\newcommand{\cororef}[1]{\hyperref[#1]{Corollary \ref{#1}}}
\newcommand{\propref}[1]{\hyperref[#1]{Proposition \ref{#1}}}
\newcommand{\remref}[1]{\hyperref[#1]{Remark \ref{#1}}}
\newcommand{\lemref}[1]{\hyperref[#1]{Lemma \ref{#1}}}
\newcommand{\clref}[1]{\hyperref[#1]{Claim \ref{#1}}}
\newcommand{\consref}[1]{\hyperref[#1]{Construction \ref{#1}}}
\newcommand{\figref}[1]{\hyperref[#1]{Figure \ref{#1}}}
\newcommand{\eqnref}[1]{\hyperref[#1]{Equation \ref{#1}}}
\newcommand{\subroutineref}[1]{\hyperref[#1]{Subroutine \ref{#1}}}

\newcommand{\ds}{\displaystyle}

\title{Approximation Algorithms for Digraph Width Parameters
\footnote{This work was done when the authors were at Georgia Institute of Technology, Atlanta, GA, USA.}}
\vspace{0.15in}
\author{
Shiva Kintali\footnote{Department of Computer Science, Princeton University, NJ, USA. Email : {\em{kintali@cs.princeton.edu}}}\ \ \ \ \ \ 
Nishad Kothari\footnote{Department of Combinatorics and Optimization, University of Waterloo, ON, Canada. Email : {\em{nkothari@math.uwaterloo.ca}}}\ \ \ \ \ \ 
Akash Kumar\footnote{Yahoo! Inc., Sunnyvale, CA, USA. Email : {\em{akashkum@yahoo-inc.com}}}
}
\date{}

\begin{document}

\maketitle

\begin{abstract}
Several problems that are {\sf NP}-hard on general graphs are efficiently solvable on graphs with bounded treewidth. Efforts have been made to generalize treewidth and the related notion of pathwidth to digraphs. Directed treewidth, DAG-width and Kelly-width are some such notions which generalize treewidth, whereas directed pathwidth generalizes pathwidth. Each of these digraph width measures have an associated decomposition structure.

In this paper, we present approximation algorithms for all these digraph width parameters. In particular, we give an $O(\sqrt{\log{n}})$-approximation algorithm for directed treewidth, and an $O(\log^{3/2}{n})$-approximation algorithm for directed pathwidth, DAG-width and Kelly-width. Our algorithms construct the corresponding decompositions whose widths are within the above mentioned approximation factors.
\end{abstract}


\section{Introduction}\label{sec:intro}
The related notions of tree decompositions and path decompositions have been studied extensively by Robertson and Seymour in their seminal work on graph minors. These decompositions correspond to associated width measures for undirected graphs called treewidth and pathwidth respectively. Besides playing a crucial role in structural graph theory, these width measures also proved to be very useful in the design of algorithms. Roughly speaking, treewidth of an undirected graph measures how close the graph is to being a tree. On the other hand, pathwidth measures how close the graph is to being a path. Several problems that are {\sf NP}-hard on general graphs are solvable in polynomial time on graphs of bounded treewidth using dynamic programming techniques. These include classical problems such as hamiltonian cycle, graph coloring, vertex cover, graph isomorphism and many more. We refer the reader to \cite{kloks}, \cite{bodlaender-survey} and references therein for an introduction to treewidth.

One attempt at solving algorithmic problems on digraphs would be to consider the treewidth of the underlying undirected graph. However, this approach suffers from certain drawbacks if the problem being considered depends on the directions of the arcs. For instance, it is possible to orient the edges of a complete graph in order to obtain a directed acyclic graph (DAG). Although the (undirected) treewidth of such a digraph is large, it is easy to solve the Hamiltonian cycle problem on such a digraph. Thus, it would be desirable to have a width measure for digraphs which would attain much lower values on such digraphs than the value of the (undirected) treewidth. Hence, efforts have been made to generalize treewidth and pathwidth to digraphs. Directed treewidth, DAG-width and Kelly-width are some such notions which generalize treewidth, whereas directed pathwidth generalizes pathwidth. Each of these digraph width measures have an associated decomposition structure as well.

Johnson et al. \cite{dtw-definition} introduced the first directed analogue of treewidth called directed treewidth. They demonstrated the algorithmic benefits of directed treewidth by providing efficient algorithms for {\sf NP}-hard problems (such as Hamiltonian cycle) on digraphs of bounded directed treewidth. Reed \cite{dtw-reed} defined another directed analogue which is closely related to the one introduced by Johnson et al. Later on, Berwanger et al. \cite{dagw-definition1} and independently Obdrzalek \cite{dagw-definition2} introduced DAG-width. They demonstrated the usefulness of DAG-width by showing that the winner of a parity game can be decided in polynomial time on digraphs of bounded DAG-width. Parity games are a certain form of combinatorial game played on digraphs. They also give an equivalent characterization of DAG-width in terms of a certain variant of the cops-and-robber game in which the robber is visible and dynamic. More recently, Hunter and Kreutzer \cite{kellywidth-definition} introduced Kelly-width. They presented several equivalent characterizations of Kelly-width such as elimination ordering, partial $k$-DAGs and another variant of the cops-and-robber game in which the robber is invisible and inert. We refer the reader to \apref{ap:cops-and-robber} for a discussion of cops-and-robber games.

All of the above mentioned width measures are generalizations of undirected treewidth. More precisely, for a graph $G$ with treewidth $k$, let $\overline{G}$ be the {\it digraph} obtained from $G$ by replacing each edge $\{u,v\}$ of $G$ by two arcs $(u,v)$ and $(v,u)$, then: $(i)$ the directed treewidth of $\overline{G}$ is equal to $k$ \cite[Theorem 2.1]{dtw-definition}, $(ii)$ the DAG-width of $\overline{G}$ is equal to $k+1$ \cite[Proposition 5.2]{dagw-definition1}, and, $(iii)$ the Kelly-width of $\overline{G}$ is equal to $k+1$ \cite{kellywidth-definition}. Similarly, directed pathwidth introduced by Reed, Seymour and Thomas is a generalization of undirected pathwidth \cite[Lemma 1]{dpw-definition}. Computing the {\it treewidth} (or {\it pathwidth}) of an undirected graph is {\sf NP}-complete \cite{tw-pw-npcomplete}. Moreover, Bodlaender et al. \cite[Theorem 23]{bodlaender-tw-approx} show that unless {\sf P}$=${\sf NP}, neither treewidth nor pathwidth can be approximated within an additive constant or term of the form $n^\epsilon$ for $\epsilon<1$ of optimal. It follows that computing any of these digraph width parameters is also {\sf NP}-complete, and futhermore a similar approximation hardness applies.

All the algorithms proposed for approximating treewidth rely on the relation between treewidth and balanced vertex separators (which we discuss in more detail shortly). In their seminal work, Leighton and Rao\cite{leighton-rao} gave an $O(\log n)$-pseudo approximation algorithm for computing balanced vertex separators. Bodlaender et al. \cite{bodlaender-tw-approx} gave an $O(\log n)$-approximation algorithm for computing treewidth. Their algorithm made use of the small vertex separators obtained using the results of \cite{leighton-rao}. Moreover, their techniques imply that any $\rho$-approximation algorithm for balanced vertex separators can be used to obtain a $\rho$-approximation algorithm for treewidth. Now, let $k$ denote the treewidth of a graph. Bouchitt{\'e} et al. \cite{bkmt-tw-approx} gave an $O(\log k)$-approximation algorithm for treewidth using different techniques. Independently, Amir \cite{eyalamir} gave another approximation algorithm with the same guarantee for treewidth, and this again relies on the algorithms of \cite{leighton-rao}.

The approximation algorithm for balanced vertex separators was improved to $O(\sqrt{\log k})$ by Feige, Hajiaghayi and Lee \cite{fhl-tw-approx}. As per the above discussion and as noted by Feige et al. \cite{fhl-tw-approx}, this gives an $O(\sqrt{\log {k}})$-approximation algorithm for treewidth. Kloks\cite{kloks} gives a procedure to transform a tree decomposition to a path decomposition whose width is at most $\log{n}$ times the width of the original tree decomposition. It follows that the result of Feige et al. \cite{fhl-tw-approx} implies an $O(\sqrt{\log k}\cdot\log{n})$-approximation algorithm for pathwidth.

To the best of our knowledge, no (non-trivial) approximation algorithms are known for any of the above mentioned digraph width parameters. We take a step in this direction. Our algorithms are similar to the above mentioned approximation algorithms for treewidth in the sense that they rely on the approximation algorithms for balanced directed vertex separators (see \defref{definition:directed-separator}). Leighton and Rao \cite{leighton-rao} observed that their algorithm can be extended to work on directed graphs as well. This leads to an $O(\log n)$-approximation algorithm for balanced directed vertex separators using the algorithm for directed edge separators as a black box. This was further improved to $O(\sqrt{\log{n}})$ by Agarwal et al. \cite{agarwal-charikar-makarychev}. Our algorithms make use of their approximation algorithm as a subroutine.

\subsection{Results and techniques}
\begin{itemize}
\item We obtain an $O(\log^{3/2}{n})$-approximation algorithm for directed pathwidth. This algorithm uses ideas similar to those of Bodlaender et al. \cite{bodlaender-tw-approx} for approximating treewidth and pathwidth, which in turn builds on techniques developed by Lagergren \cite{tw-lagergren-algo} and Reed \cite{tw-reed-algo}. Let $G$ be an undirected graph. Informally speaking, a balanced vertex separator is a set of vertices $S\subseteq V(G)$ such that $V(G)-S$ can be divided into two parts of roughly the same size. Their algorithm at a high level uses a divide-and-conquer approach to compute approximate path decompositions of the graphs induced by these two parts, and then uses these to construct an approximate path decomposition of $G$. We refer the reader to \cite[Section 6.1]{kloks} for a detailed description of this algorithm. The approximation guarantee of their algorithm crucially depends on the fact that every graph of treewidth $k$ has a balanced vertex separator of size at most $k+1$. We first establish analogous relations between balanced directed vertex separators (see \defref{definition:directed-separator}) and all of the relevant digraph width parameters, and then use a similar divide-and-conquer approach to compute an approximate directed path decomposition.
\item It turns out that our approximation algorithm for directed pathwidth can also be used to approximate DAG-width (see \defref{definition:DAG-decomposition}) in a natural way. Thereafter, we formulate a width parameter called Kelly pathwidth and the associated decomposition called Kelly path decomposition, and show that these are essentially equivalent to directed pathwidth and directed path decomposition respectively. Although this is not surprising, it turns out to be useful to show how our approximation algorithm for directed pathwidth can also be used to approximate Kelly-width (see \defref{definition:Kelly-decomposition}).
\item We obtain an $O(\sqrt{\log{n}})$-approximation algorithm for directed treewidth (see \defref{definition:arboreal-decomposition}). This algorithm also uses a divide-and-conquer approach, and computes balanced directed vertex separators. However, in this case the divide step may lead to more than two subproblems. Once the approximate arboreal decompositions for these subproblems have been computed, we combine these to construct a decomposition of the input digraph while ensuring that all the conditions stated in \defref{definition:arboreal-decomposition} are met. This construction is based on the proof of Theorem 3.3 of Johnson et al. \cite{dtw-definition} which basically states that a digraph with a {\it haven} of large order has a large directed treewidth. Havens correspond to certain winning strategies in a variant of the cops-and-robber game introduced in \cite{tw-visible-dynamic}. Once again, the approximation guarantee relies on the relations established between balanced directed vertex separators and directed treewidth.
\end{itemize}
\subsection{Organization of this paper}
In \secref{sec:notation-terminology} we describe some notation and terminology used throughout the paper. In \secref{sec:digraphwidthparameters} we discuss the formal definitions of the relevant digraph width parameters. Apart from that, we formulate a new width parameter called Kelly pathwidth and show that it is equivalent to directed pathwidth. In \secref{sec:approximation} we give approximation algorithms for all of these digraph width parameters. In particular, \secref{sec:separator approximation} discusses the notion of balanced directed vertex separators and proves some bounds which are essential for the approximation guarantees of our algorithms. Thereafter, \secref{sec:dpw approximation} presents an $O(\log^{3/2}{n})$-approximation algorithm for directed pathwidth, DAG-width and Kelly-width and \secref{sec:approx-dtw} presents an $O(\sqrt{\log{n}})$-approximation algorithm for directed treewidth. Finally, we discuss some open problems and directions for future work in \secref{sec:discussion}.

\section{Notation and terminology}\label{sec:notation-terminology}
\indent We use standard graph theory notation and terminology (see \cite{diestel}). All digraphs are finite and simple (i.e. no self loops and no multiple arcs). We use the term DAG when referring to directed acyclic graphs. For a digraph $G$, we write $V(G)$ for its vertex set and $E(G)$ for its arc set. For $S \subseteq V(G)$ we write $G[S]$ for the subdigraph induced by $S$, and $G \setminus S$ for the subdigraph induced by $V(G) - S$.

\indent Let $T$ be a DAG. For two distinct nodes $i$ and $j$ of $T$, we write $i \prec_T j$ if there is a directed walk in $T$ with first node $i$ and last node $j$. For convenience, we write $i \prec j$ whenever $T$ is clear from the context. For nodes $i$ and $j$ of $T$, we write $i \preceq j$ if either $i=j$ or $i \prec j$. For an arc $e=(i,j)$ and a node $k$ of $T$, we write $e \prec k$ if either $j=k$ or $j \prec k$. We write $e \sim i$ (and $e \sim j$) to mean that $e$ is incident with $i$ (and $j$ respectively).\\
\indent Let $\mathcal{W}=(W_i)_{i \in V(T)}$ be a family of finite sets called {\it node bags}, which associates each node $i$ of $T$ to a node bag $W_i$. Let $\mathcal{X}=(X_e)_{e \in E(T)}$ be a family of finite sets called {\it arc bags}, which associates each node $i$ of $T$ to an arc bag $X_e$. We write $W_{\succeq i}$ to denote $\ds\bigcup_{j \succeq i} W_j$, and $X_{\sim i}$ to denote $\ds\bigcup_{e \sim i}X_e$. For an arc $e$ of $T$, we write $W_{\succ e}$ to denote $\ds\bigcup_{j \succ e} W_j$.\\
\indent A node is a {\it root} if it has no incoming arcs, and it is a {\it sink} if it has no outgoing arcs. The DAG $T$ is an {\it arborescence} if it has a unique root $r$ such that for every node $i \in V(T)$ there is a unique directed walk from $r$ to $i$. Note that every arborescence arises from an undirected tree by selecting a root and directing all edges away from the root. The DAG $T$ is a {\it directed path graph} if it is an arborescence whose underlying undirected graph is a path.\\
\indent Now, let $G$ be a digraph. Width measures like DAG-width and Kelly-width are based on the following notion of {\it guarding}:
\begin{definition}[Guarding]\label{definition:guarding}
Let $W, X \subseteq V(G)$. We say $X$ {\it guards} $W$ if $W \cap X = \emptyset$, and for all $(u,v) \in E(G)$, if $u \in W$ then $v \in W \cup X$.
\end{definition}

In other words, $X$ guards $W$ means that there is no directed path in $G \setminus X$ that starts from $W$ and leaves $W$. The notion of directed treewidth is based on a weaker condition:
\begin{definition}[$X$-normal]\label{definition:normal}
Let $W, X \subseteq V(G)$. We say $W$ is {\it $X$-normal} if $W \cap X = \emptyset$, and there is no directed path in $G \setminus X$ with first and last vertices in $W$ that uses a vertex of $G \setminus (W \cup X)$.
\end{definition}

In other words, $W$ is $X$-normal means that there is no directed path in $G \setminus X$ that starts from $W$, leaves $W$ and then returns to $W$. The following is a relevant observation:
\begin{rem}\label{rem:normal}
$W$ is $X$-normal if and only if the vertex-sets of the strongly connected components of $G \setminus X$ can be enumerated as $W_1, W_2, ..., W_k$ in such a way that:
\begin{itemize}
\item if $1 \le i < j \le k$, then no edge of $G$ has head in $W_i$ and tail in $W_j$, and
\item either $W = \emptyset$, or $W = W_i \cup W_{i+1} \cup ... \cup W_j$ for some integers $i,j$ with $1 \le i \le j \le k$.
\end{itemize} 
\end{rem}

\section{Digraph Width Parameters}\label{sec:digraphwidthparameters}

\subsection{DAG-width and directed pathwidth}\label{sec:dgw-dpw}

DAG-decomposition and DAG-width were introduced by Berwanger et al. \cite{dagw-definition1}, and independently by Obdr\v{z}\'{a}lek \cite{dagw-definition2}.

\begin{definition}[DAG-decomposition and DAG-width \cite{dagw-definition1}\cite{dagw-definition2}]\label{definition:DAG-decomposition}
A {\it DAG-decomposition} of a digraph $G$ is a pair $\mathcal{D} = (T,\mathcal{W})$ where $T$ is a DAG, and $\mathcal{W} = (W_i)_{i \in V(T)}$ is a family of subsets (node bags) of $V(G)$, such that:
\begin{itemize}
\item $\bigcup_{i \in V(T)} W_i = V(G)$. \hfill{\rm{\sf (DGW-1)}}
\item For all nodes $i,j,k \in V(T)$, if $i \preceq j \preceq k$, then $W_i \cap W_k \subseteq W_j$. \hfill{\rm{\sf (DGW-2)}}
\item For all arcs $(i,j) \in E(T)$, $W_i \cap W_j$ guards $W_{\succeq j} \backslash W_i$. For any root $r \in V(T)$,\\
$W_{\succeq r}$ is guarded by $\emptyset$. \hfill{\rm{\sf (DGW-3)}}
\end{itemize}
The width of a DAG-decomposition $\mathcal{D}=(T,\mathcal{W})$ is defined as $\max\{|W_i|:i \in V(T)\}$. The {\it DAG-width} of $G$, denoted by $dgw(G)$, is the minimum width over all possible DAG-decompositions of $G$.
\end{definition}

In order to define {\it directed path decomposition} and {\it directed pathwidth}, we restrict the underlying decomposition $T$ to be a directed path graph in \defref{definition:DAG-decomposition}. We denote the directed pathwidth of $G$ by $dpw(G)$.\\
\indent Directed path decomposition and directed pathwidth were introduced by Reed, Seymour and Thomas \cite{dpw-definition} as a generalization of pathwidth to digraphs. In their definition they subtracted one from the width, which is consistent with the definition of pathwidth of undirected graphs. Next we show that one may replace the condition {\sf (DGW-3)} by the following equivalent condition:\\

\noindent {\it For all arcs $(u,v) \in E(G)$, there exist nodes $i, j \in V(T)$ such that $i \preceq j$, $u \in W_i$, and, $v \in W_j$. \hfill{\sf (DPW)}}\\

\begin{prop}\label{prop:dgw3-equivalentto-dpw}
Let $G$ be a digraph. Let $\mathcal{D}=(T,\mathcal{W})$ where $T$ is a directed path graph and $\mathcal{W}=(W_i)_{i \in V(T)}$ is a family of subsets (node bags) of $V(G)$, such that it satisifies conditions {\sf (DGW-1)} and {\sf (DGW-2)}. Then $\mathcal{D}$ satisfies condition {\sf (DGW-3)} if and only if it satisfies condition {\sf (DPW)}.
\end{prop}
\begin{proof}
Say $V(T)=\{1, ..., l\}$, where the arcs are $(i,i+1)$ for each $i \in \{1,...,l-1\}$. Suppose $\mathcal{D}$ satisfies condition {\sf (DGW-3)}. We show that $\mathcal{D}$ satisifies condition {\sf (DPW)} as well, i.e. for each $(u,v)$ in $E(G)$, there exist nodes $i,j$ in $V(T)$ such that $i \le j$, $u$ appears in $W_i$, and $v$ appears in $W_j$. Suppose not. Let $(u,v)$ be an arc of $G$ that violates this condition. Let $i$ be the largest index such that $v$ appears in $W_i$, and $j$ be the smallest index such that $u$ appears in $W_j$. It follows that $i < j$. Note that $v$ does not lie in $W_{j-1} \cap W_j$. Also, $W_{\ge j} \setminus W_{j-1}$ contains $u$ but not $v$. The existence of arc $(u,v)$ in $G$ implies that $W_{j-1} \cap W_j$ does not guard $W_{\ge j}\setminus W_i$, contradicting our assumption that $\mathcal{D}$ satisfies condition {\sf (DGW-3)}.\\
\indent Now, suppose $\mathcal{D}$ satisfies condition {\sf (DPW)}. We show that $\mathcal{D}$ satisfies condition {\sf (DGW-3)} as well, i.e. for each arc $(i,i+1)$ in $E(T)$, $W_i \cap W_{i+1}$ guards $W_{\ge i+1} \setminus W_i$, and that $W_{\ge 1}$ is guarded by $\emptyset$. Suppose not. Observe that $W_{\ge 1} = V(G)$ (due to {\sf (DGW-1)}) is trivially guarded by $\emptyset$. Let $(i,i+1)$ be an arc of $T$ such that $W_i \cap W_{i+1}$ does not guard $W_{\ge i+1} \setminus W_i$. From \defref{definition:guarding}, this implies that there exists an arc $(u,v) \in E(G)$ such that $u$ lies in $W_{\ge i+1} \setminus W_i$, and $v$ does not lie in $(W_{\ge i+1} \setminus W_i) \cup (W_i \cap W_{i+1})$. It follows that the largest index $j$ such that $v$ appears in $W_j$ must satisfy $j \le i$. Also, $i$ is the smallest index such that $u$ appears in $W_i$ (due to {\sf (DGW-2)}). We conclude that the existence of arc $(u,v)$ in $G$ violates the condition {\sf (DPW)}, contradicting our assumption.
\end{proof}
It follows from the definitions that a directed path decomposition is also a DAG-decomposition. In this manner, DAG-width generalizes directed pathwidth:
\begin{prop}\label{prop:dpw dgw relation}\footnote{This was shown previously by Berwanger et al.\cite{dagw-journal} using the original definition of directed pathwidth.}
For a digraph $G$, $dgw(G) \le dpw(G)$.
\end{prop}

\subsection{Kelly-width and Kelly pathwidth}
\indent Kelly-decomposition and Kelly-width were introduced by Hunter and Kreutzer \cite{kellywidth-definition}.
\begin{definition}[Kelly-decomposition and Kelly-width \cite{kellywidth-definition}]\label{definition:Kelly-decomposition}
A {\it Kelly-decomposition} of a digraph $G$ is a triple $\mathcal{D} = (T, \mathcal{W}, \mathcal{X})$ where $T$ is a DAG, and $\mathcal{W}=(W_i)_{i\in V(T)}$ and $\mathcal{X}=(X_i)_{i\in V(T)}$ are families of subsets (node bags) of $V(G)$, such that:
\begin{itemize}
\item $\mathcal{W}$ is a partition of $V(G)$. \hfill{\rm{\sf (KW-1)}}
\item For all nodes $i \in V(T), X_i$ guards $W_{\succeq i}$. \hfill{\rm{\sf (KW-2)}}
\item For each node $i \in V(T)$, the children of $i$ can be enumerated as $j_1, ... , j_s$ so that for each $j_q$, $X_{j_q} \subseteq W_i \cup X_i \cup \bigcup_{p<q} W_{\succeq j_p}$. Also, the roots of $T$ can be enumerated as $r_1, r_2, ...$ such that for each root $r_q$, $W_{r_q} \subseteq \bigcup_{p<q} W_{\succeq r_p}$. \hfill{\rm{\sf (KW-3)}}
\end{itemize}
The width of a Kelly-decomposition $\mathcal{D}=(T,\mathcal{W},\mathcal{X})$ is defined as $\max\{|W_i \cup X_i| : i \in V(T)\}$. The {\it Kelly-width} of $G$, denoted by $kw(G)$, is the minimum width over all possible Kelly-decompositions of $G$.
\end{definition}

In order to define {\it Kelly path decomposition} and {\it Kelly pathwidth}, we restrict the underlying decomposition $T$ to be a directed path graph in \defref{definition:Kelly-decomposition}. We denote the Kelly pathwidth of $G$ by $kpw(G)$. In this case, the condition {\sf (KW-3)} simplifies to:\\

\noindent {\it For all arcs $(i,j) \in E(T)$, $X_{j} \subseteq W_i \cup X_i$. \hfill {\sf (KPW)}}\\

\indent It follows from the definitions that a Kelly path decomposition is also a Kelly-decomposition. In this manner, Kelly-width generalizes Kelly pathwidth:
\begin{prop}\label{prop:kw-kpw}
For a digraph $G$, $kw(G) \le kpw(G)$.
\end{prop}
Now, we show that for a digraph $G$ its Kelly pathwidth equals its directed pathwidth. The high level idea of our proof is as follows. Given a directed path decomposition of $G$, there is a natural way to construct a Kelly path decomposition of $G$ which has the same width, and vice versa.
\begin{theorem}\label{thm:kpw-equals-dpw}
For any digraph $G$, $kpw(G) = dpw(G)$.
\end{theorem}
\begin{proof}
\indent First we show that $kpw(G) \le dpw(G)$. Suppose $\mathcal{D}=(T,\mathcal{W})$ is a {\it directed path decomposition} of $G$. Say $V(T)=\{1, ..., l\}$, where the arcs are $(i,i+1)$ for each $i \in \{1,...,l-1\}$. Note that, if for some $1 \le i \le l-1$ we have $W_{i+1} \subseteq W_i$, then we can delete node $i$ from $T$ and add an arc $(i-1,i+1)$ to get a directed path graph $T_1$. Let $\mathcal{W}_1$ be the restriction of $\mathcal{W}$ to $V(T_1)$. It is easy to see that $(T_1, \mathcal{W}_1)$ is a directed path decompostion of $G$, of width no more than that of $\mathcal{D}$. Thus, we may assume that $W_{i+1} \nsubseteq W_i$ for all $1 \le i < l$. Now, we describe a Kelly path decomposition $\mathcal{D'}=(T,\mathcal{W}',\mathcal{X}')$ such that width of $\mathcal{D'}$ is the same as that of $\mathcal{D}$. We set $\mathcal{W}'=(W'_i)_{i \in V(T)}$ and $\mathcal{X}'=(X'_i)_{i \in V(T)}$ as follows:
\begin{itemize}
\item $W'_1 := W_1$, and for each $i \in \{2, ..., l\}$, $W'_i := W_i \setminus W_{i-1}$.
\item $X'_1 := \emptyset$, and for each $i \in \{2, ..., l\}$, $X'_i := W_i \cap W_{i-1}$.
\end{itemize}
Observe that $\mathcal{D}'$ satisifes conditions {\sf (KW-1)} and {\sf (KPW)}.\\
\indent We now show that $\mathcal{D}'$ satisfies condition {\sf (KW-2)}, i.e. for each $i \in V(T)$, $X'_i$ guards $W'_{\ge i}:=\bigcup_{k \ge i} W_k$. Suppose not. Then for some $i \in V(T)$, there is an arc $(u,v) \in E(G)$ such that $u \in W'_{\ge i}$ and $v \notin X'_i \cup W'_{\ge i}$. Note that $u$ is contained in $W'_j$ where $j \ge i$ is the smallest integer such that $u \in W_j$. Since $\mathcal{D}'$ satisfies condition {\sf (KPW)}, $v \notin X'_k$ for any $k \ge i$. Hence, the largest integer $k$ such that $v \in W_k$ must satisfy $k < i$. However, this violates condition {\sf (DPW)} for the arc $(u,v)$. Hence, $\mathcal{D}'$ is a Kelly path decomposition of $G$.\\
\indent Next we show that $dpw(G) \le kpw(G)$. Suppose $\mathcal{D}'=(T, \mathcal{W}', \mathcal{X}')$ is a {\it Kelly path decomposition} of $G$. Say $V(T)=\{1,...,l\}$, where the arcs are $(i,i+1)$ for each $i \in \{1,...,l-1\}$. Now, we describe a directed path decomposition $\mathcal{D}=(T,\mathcal{W})$ such that width of $\mathcal{D}$ is the same as that of $\mathcal{D}'$. We set $\mathcal{W}=(W_i)_{i \in V(T)}$ such that for each $i \in V(T)$, $W_i:=W'_i\cup X'_i$. Observe that $\mathcal{D}$ satisfies condition {\sf (DGW-1)}.\\
\indent We now show that $\mathcal{D}$ satisfies condition {\sf (DGW-2)}, i.e. for $i < j < k$, if $v$ lies in $W_i \cap W_k$, then $v$ must appear in $W_j$. Note that $W_i \cap W_k=(W'_i \cup X'_i) \cap (W'_k \cup X'_k) = (W'_i \cap X'_k) \cup (X'_i \cap X'_k)$, where the latter equality follows since $W'_i \cap W'_k = \emptyset$ (due to {\sf (KPW-1)}), and $X'_i \cap W'_k = \emptyset$ (due to {\sf (KPW-2)}). It follows that if $v$ lies in $W_i \cap W_k$, then $v$ also lies in $X'_k$. It follows from {\sf (KPW-3)} that $v$ lies in $W'_{k-1} \cup X'_{k-1}$. Also, $v$ must appear in one of $W'_i$ and $X'_i$. We consider these cases separately. If $v$ appears in $W'_i$, then $v$ does not appear in $W'_{k-1}$ (due to {\sf (KPW-1)}). Otherwise, $v$ appears in $X'_i$. Again, $v$ can not appear in $W'_{k-1}$ since $X'_i$ guards $W'_{\ge i}$. Hence, in either case, $v$ must appear in $X'_{k-1}$. It follows that $v$ lies in $(W'_i \cap X'_{k-1}) \cup (X'_i \cap X'_{k-1})$. Applying this argument repeatedly, we conclude that $v$ lies in $X'_j$. This implies that $v$ appears in $W_j$.\\
\indent We now show that $\mathcal{D}$ satisfies condition {\sf (DPW)}, i.e. for each $(u,v)$ in $E(G)$, there exist nodes $i,j$ in $V(T)$ such that $i \le j$, $u$ appears in $W_i$, and $v$ appears in $W_j$. Suppose not. Let $(u,v)$ be an arc of $G$ that violates this condition. Let $i$ be the largest index such that $v$ appears in $W_i$, and $j$ be the smallest index such that $u$ appears in $W_j$. It follows that $i < j$. This means that $u$ lies in $W'_j \cup X'_j$. If $u$ lies in $X'_j$, then by {\sf (KPW-1)} and {\sf (KPW-2)}, we conclude that $u$ lies in $W'_k$ for some $k < j$. But this contradicts our choice of $j$. Thus, $u$ must lie in $W'_j$. Note that $v$ lies neither in $W'_{\ge j}$, nor in $X'_j$. This implies that the arc $(u,v)$ violates the condition {\sf (KPW-2)} for node $j$. Hence, $\mathcal{D}$ is a {\it directed path decomposition} of $G$. This completes the proof.
\end{proof}

It follows that Kelly-width generalizes directed pathwidth. Next, we show that the gap between Kelly pathwidth and Kelly-width can be arbitrarily large. Berwanger et al. \cite[Proposition 36]{dagw-journal} show a family of digraphs with arbitrarily large directed pathwidth and DAG-width $2$. It is easy to show that this family of digraphs gives an analogous result for Kelly pathwidth and Kelly-width. However, this relies on the notion of cops-and-robber games. We provide the details in \apref{ap:cops-and-robber}.

\begin{prop}\label{prop:kpw-kw-gap}
There exist a family of digraphs with arbitrarily large Kelly pathwidth and Kelly-width $2$.
\end{prop}

\subsection{Directed treewidth}
Arboreal decomposition and directed treewidth were introducted by Johnson et al. \cite{dtw-definition}.
\begin{definition}[Arboreal decomposition and directed treewidth\cite{dtw-definition}]\label{definition:arboreal-decomposition}
An {\it arboreal decomposition} of a digraph $G$ is a triple $\mathcal{D} = (T, \mathcal{W}, \mathcal{X})$, where $T$ is an arborescence, and $\mathcal{W}=(W_i)_{i\in V(T)}$ is a family of subsets (node bags) of $V(G)$, and $\mathcal{X}=(X_e)_{e\in E(T)}$ is a family of subsets (arc bags) of $V(G)$, such that:
\begin{itemize}
\item { $\mathcal{W}$ is a partition of $V(G)$. \hfill{\rm{\sf (DTW-1)}}  }
\item { For each arc $e \in E(T)$, $W_{\succ e}$ is $X_e$-normal. \hfill{\rm{\sf (DTW-2)}} }
\end{itemize}
The width of an arboreal decomposition $\mathcal{D}=(T,\mathcal{W},\mathcal{X})$ is defined as $\max\{|W_i \cup X_{\sim i}|:i \in V(T)\} - 1$. The {\it directed treewidth} of $G$, denoted by $dtw(G)$, is the minimum width over all possible arboreal decompositions of $G$.
\end{definition}

In \secref{sec:approx-dtw}, we will construct an arboreal decomposition incrementally. For this purpose, given a set $U \subseteq V(G)$, we define an {\it arboreal decomposition with respect to $U$}. To do so, replace condition {\sf (DTW-1)} by: $\mathcal{W}$ is a partition of $U$ (in \defref{definition:arboreal-decomposition}). Note that when $U=V(G)$, this is an arboreal decomposition of $G$.
\begin{rem}\label{rem:trivial}
For a digraph $G$, consider the decomposition $\mathcal{D}$ such that the underlying DAG $T$ has a single node, whose corresponding node bag is $V(G)$. Note that $\mathcal{D}$ satisfies the properties of each digraph decomposition mentioned above. We refer to such a decomposition as the {\it trivial decomposition} of $G$. Moreover, if the unique node bag is $U$ for some $U \subseteq V(G)$, then the decomposition is a {\it trivial arboreal decomposition} of $G$ {\it with respect to $U$}.
\end{rem}

\section{Approximation Algorithms}\label{sec:approximation}
\subsection{Balanced separators}\label{sec:separator approximation}
Our algorithms for approximating digraph width parameters are similar to earlier work of Bodlaender et al. \cite{bodlaender-tw-approx} for approximating undirected width parameters. They show that treewidth has a useful relation with balanced (undirected) vertex separators, and exploit this relation to obtain approximation algorithms for treewidth as well as pathwidth.\\
\indent We establish similar relationships between digraph width parameters and balanced directed vertex separators. To do so, we need some definitions.
\begin{definition}[$\alpha$-balanced directed vertex separator]\label{definition:directed-separator}
Let $G$ be a digraph and $U \subseteq V(G)$. Let $\alpha \in (0,1)$. An $\alpha$-balanced directed vertex separator of $U$ is a set $S \subseteq V(G)$ such that $V(G) - S$ can be partitioned into two sets $U_1$ and $U_2$, each of which has at most $\alpha \cdot |U|$ vertices of $U$, and such that $S$ guards $U_2$ in $G$.
\end{definition}

In the above definition, when $U=V(G)$, we refer to $S$ as an $\alpha$-balanced directed vertex separator of $G$.
\begin{definition}[$\alpha$-directed separator number]
For a digraph $G$, the $\alpha$-directed separator number of $G$, denoted by $dsn_\alpha(G)$, is the smallest $k$ such that every subset of $V(G)$ has an $\alpha$-balanced directed vertex separator of size no larger than $k$.
\end{definition}

The next two results show that for a digraph $G$ it is possible to obtain a lower bound for each of its digraph width parameters (discussed in \secref{sec:digraphwidthparameters}) in terms of its $\frac{3}{4}$-directed separator number.
\begin{prop}\label{prop:dtw separator property}
Let $G$ be a digraph whose directed treewidth is $k$. For any $U \subseteq V(G)$, there exists a $\frac{3}{4}$-balanced directed vertex separator of $U$ of size at most $k+1$.
\end{prop}
\begin{proof}
Let $\mathcal{D}=(T,\mathcal{W},\mathcal{X})$ be an arboreal decomposition of $G$, whose width is $k$. Let $r$ be the root of $T$. Pick the unique node $q$ such that it satisfies the property that $W_{\succeq q}$ contains at least $\frac{1}{2}|U|$ vertices of $U$, and the distance between $q$ and $r$ is maximized. Consider the set $S:=W_q \cup X_{\sim q}$. Let $C_1, C_2, ..., C_l$ be the strongly connected components of $G-S$ sorted in topological order. It follows from the choice of $q$ and \remref{rem:normal} that each $C_i$ has at most $\frac{1}{2}|U|$ vertices of $U$. Note that $S$ is of size at most $k+1$. Now, it suffices to show that $S$ is a $\frac{3}{4}$-balanced directed vertex separator of $U$ in $G$. To do so, we group the vertex sets of $C_1, C_2, ..., C_l$ into two sets $U_1$ and $U_2$ (see \defref{definition:directed-separator}) as follows:\\
{\it Case $1:$} $V(G)-S$ contains at most $\frac{3}{4}|U|$ vertices of $U$.\\
Set $U_1:=V(G)-S$, and $U_2:=\emptyset$.\\
{\it Case $2:$} Some $C_i$ contains at least $\frac{1}{4}|U|$ vertices of $U$.\\
Let $|V(C_i) \cap U| = (\frac{1}{4}+\theta)|U|$, where $0 \le \theta \le \frac{1}{4}$. Set $A:= \bigcup_{j=1}^{i-1} C_j$ and $B:=\bigcup_{j=i+1}^{l} C_j$. It follows that $|A \cap U| + |B \cap U| \le (\frac{3}{4} - \theta)|U|$. If $|A \cap U| \le (\frac{3}{8}-\frac{\theta}{2})|U|$. Set $U_1:=A \cup V(C_i)$ and $U_2:=B$. Otherwise, $|B \cap U| \le (\frac{3}{8}-\frac{\theta}{2})|U|$ must hold true. Set $U_1:=A$ and $U_2:=V(C_i) \cup B$. It can be verified that both $U_1$ and $U_2$ contain at most $\frac{3}{4}|U|$ vertices of $U$.\\
{\it Case $3:$} Each $C_i$ contains strictly fewer than $\frac{1}{4}|U|$ vertices of $U$.\\
Let $1 \le j \le l$ be such that $\sum_{i=1}^{j-1} |C_i \cap U| < \frac{1}{4}|U|$ and $\sum_{i=1}^j |C_i \cap U| \ge \frac{1}{4}|U|$. Set $U_1:=\bigcup_{i=1}^j V(C_i)$ and $U_2:=\bigcup_{i=j+1}^l V(C_i)$. It can be easily verified that $U_1$ contains at most $\frac{1}{2}|U|$ vertices of $U$, and $U_2$ contains at most $\frac{3}{4}|U|$ vertices of $U$.
\end{proof}
\begin{cor}\label{cor:dtw separator property}
For a digraph $G$, (i) $dsn_\frac{3}{4}(G)-1 \le dtw(G)$, (ii) $dsn_\frac{3}{4}(G) -2 \le 3 \cdot dgw(G)$, (iii) $dsn_\frac{3}{4}(G) + 1 \le 6 \cdot kw(G)$, and, (iv) $dsn_\frac{3}{4}(G)-2 \le 3 \cdot dpw(G)$.
\end{cor}
\begin{proof}
\propref{prop:dtw separator property} implies (i). It follows from \cite[Proposition 35]{dagw-journal} that $dtw(G) \le 3 \cdot dgw(G) +1$. This, along with (i) proves (ii). Now, \propref{prop:dpw dgw relation} leads to (iv). It follows from \cite[Corollary 21]{kellywidth-definition} that $dtw(G) \le 6 \cdot kw(G)-2$. This, along with (i) proves (iii).
\end{proof}
The bounds established in \cororef{cor:dtw separator property} are crucial in proving the approximation guarantees of our algorithms. This is due to the fact that our algorithms use the approximation algorithms for balanced directed vertex separators as a subroutine. Leighton and Rao \cite{leighton-rao} presented an $O(\log n)$ pseudo-approximation algorithm for computing balanced directed vertex separators. This was improved to $O(\sqrt{\log n})$ by Agarwal et al. \cite{agarwal-charikar-makarychev}.
\begin{theorem}\label{theorem:acmm-directed-vertex-separator}
There exists a polynomial time approximation algorithm that, given a digraph $G$, $U \subseteq V(G)$ and parameter $\alpha \in [\frac{1}{2},1)$, finds an $\alpha'$-balanced directed vertex separator of $U$ of size $O(\sqrt{\log n}\cdot dsn_\alpha(G))$ for any $\alpha'$ such that $\alpha' > \alpha$ and $\alpha' \ge \frac{2}{3}$.
\end{theorem}

Plugging the value of $\alpha$ as $\frac{3}{4}$ in \thref{theorem:acmm-directed-vertex-separator}, and using inequality (i) from \cororef{cor:dtw separator property}, we get the following result:
\begin{cor}\label{cor:separator algorithms}
There exists a constant $\beta$ and a polynomial time approximation algorithm, call it {\sf FindSep}, that given a digraph $G$ and $U \subseteq V(G)$, finds an $\alpha'$-balanced directed vertex separator of $U$ of size at most $\beta\sqrt{\log{n}} \cdot dtw(G)$, for any $\alpha' > \frac{3}{4}$.
\end{cor}

We write $(S;U_1,U_2):={\rm{\sf FindSep}}(G,U,\alpha')$ to denote that $S$ is the computed $\alpha'$-balanced directed vertex separator of $U$ in $G$, and $U_1$ and $U_2$ are the two parts of $V(G)-S$ (see \defref{definition:directed-separator}). If $U_1$ and $U_2$ are irrelevant, we simply write $S:={\rm{\sf FindSep}}(G,U,\alpha')$.
\begin{rem}\label{rem:separator algorithms}
Note that, in \cororef{cor:separator algorithms}, one can replace directed treewidth by any other digraph width parameter mentioned in \secref{sec:digraphwidthparameters}. The proof follows by using the appropriate inequality from \cororef{cor:dtw separator property}.
\end{rem}

\subsection{Approximating directed pathwidth, DAG-width and Kelly-width}\label{sec:dpw approximation}

In this section, we first present an $O(\log^{3/2}{n})$-approximation algorithm for directed pathwidth. Next, we explain how essentially the same algorithm works as an $O(\log^{3/2}{n})$-approximation algorithm for both DAG-width and Kelly-width.\\
\indent Our algorithm uses a divide-and-conquer approach whose high level idea is based on the following observation. Let $G$ be a digraph and let $S$ be a directed vertex separator of $G$. Let $U_1$ and $U_2$ be the two parts of $V(G)-S$ (see \defref{definition:directed-separator}). Let $\mathcal{D}_1= (T_1, \mathcal{W}_1)$ and $\mathcal{D}_2= (T_1, \mathcal{W}_1)$ be directed path decomposition of $G[U_1]$ and $G[U_2]$ respectively. We now describe how one can obtain a directed path decomposition of $G$ by merging the decompositions $\mathcal{D}_1$ and $\mathcal{D}_2$.
\begin{sub}[Merge]\label{sub:merge}
We write $\mathcal{D} := {\rm{\sf Merge}}(\mathcal{D}_1,\mathcal{D}_2;S)$ when the decomposition $\mathcal{D}=(T,\mathcal{W})$ is constructed as follows: $T$ is obtained by taking the union of $T_1$ and $T_2$ and adding an edge from the unique sink of $T_1$ to the unique root of $T_2$. For each node $i$ of $T$, the node bag $W_i$ is defined as the union of $S$ and the node bag at $i$ with respect to decomposition $\mathcal{D}_1$ or $\mathcal{D}_2$ (as applicable).
\end{sub}
\begin{clm}\label{clm:merge}
$\mathcal{D}={\rm{\sf Merge}}(\mathcal{D}_1,\mathcal{D}_2;S)$ is a directed path decomposition of $G$.
\end{clm}
\begin{proof}
Note that $U_1$ and $U_2$ are disjoint, and $V(G) = U_1 \cup U_2 \cup S$. It follows from the construction of $\mathcal{D}$ that it satisfies the conditions {\sf (DGW-1)} and {\sf (DGW-2)}.\\
\indent Let $(u,v)$ be any arc of $G$. It suffices to show that the condition {\sf (DPW)} is satisfied for $(u,v)$. This is easy to see if both $u, v$ lie in either $U_1$ or $U_2$ or $S$. Note that the vertices in $S$ appear in every node bag of $\mathcal{D}$. Since $S$ guards $U_2$, it follows that if $u \in U_2$, then $v$ must lie in $S \cup U_2$. Thus, {\sf (DPW)} is satisfied for these arcs as well. By the same reasoning, we conclude that {\sf (DPW)} is satisfied if $u \in U_1$ and $v \in S$. Now, $T$ is constructed by adding an arc from the sink of $T_1$ to the root of $T_2$. Hence, {\sf (DPW)} is satisfied if $u \in U_1$ and $v \in U_2$.
\end{proof}

The {\sf Merge} subroutine suggests a natural divide-and-conquer algorithm. We now formally describe the recursion subroutine {\sf MakeDPDec}, which receives a single argument $G[U]$ where $U \subseteq V(G)$ and returns a directed path decomposition of $G[U]$.\\

\noindent {\sf MakeDPDec}$(G[U])$
\begin{enumerate}
\item {\it Termination Step:}\footnote{Refer to \cororef{cor:separator algorithms}} If $|U| \le \beta\log^{3/2} n$, return the trivial decomposition of $G[U]$.
\item {\it Divide Step:} Let $(S;U_1,U_2):={\rm{\sf FindSep}}(G[U],U,\alpha')$.\\Recursively compute $\mathcal{D}_1:={\rm{\sf MakeDPDec}}(G[U_1])$ and $\mathcal{D}_2:={\rm{\sf MakeDPDec}}(G[U_2])$.
\item {\it Combine Step:} Let $\mathcal{D}:={\sf{\rm Merge}}(\mathcal{D}_1,\mathcal{D}_2;S)$. Return $\mathcal{D}$.
\end{enumerate}

Given a digraph $G$, our algorithm is just {\sf MakeDPDec}$(G)$. We fix $\alpha' \in (\frac{3}{4},1)$ throughout the algorithm.
\begin{lem}
Given a digraph $G$, {\sf MakeDPDec}$(G)$ returns a directed path decomposition of $G$, whose width is $O(\log^{3/2}n\cdot dpw(G))$.
\end{lem}
\begin{proof}
It follows from \remref{rem:trivial} that at the termination step, the algorithm returns a directed path decomposition of the input graph. Let $\mathcal{D}$ be the decomposition returned by {\sf MakeDPDec}$(G)$. By repeated application of \clref{clm:merge}, it follows that $\mathcal{D}$ is a directed path decomposition of $G$. It remains to show that the size of each node bag is $O(\log^{3/2}n\cdot dpw(G))$.\\
\indent Note that each invocation of {\sf MakeDPDec} finds an $\alpha'$-balanced directed vertex separator. This guarantees that the depth of the recursion tree is $O(\log n)$. Thus, each node bag of $\mathcal{D}$ comprises of $c\log n$ separators (where $c$ is a constant whose value depends only on $\alpha'$). Using \cororef{cor:separator algorithms} and \remref{rem:separator algorithms}, each separator is of size at most $\beta\sqrt{\log n}\cdot dpw(G)$. The termination step ensures that the size of the node bag in the trivial decomposition is at most $\beta\log^{3/2} n$. It follows that the size of each node bag is at most $c\log n \cdot \beta\sqrt{\log{n}} \cdot dpw(G) + \beta\log^{3/2}n = O(\log^{3/2} n \cdot dpw(G))$.
\end{proof}
\begin{theorem}\label{thm:dpd-algo}
There exists a polynomial time approximation algorithm that, given a digraph $G$, computes a directed path decomposition of $G$, whose width is $O(\log^{3/2}n \cdot dpw(G))$.
\end{theorem}

Note that a directed path decomposition is also a DAG-decomposition. It follows from \remref{rem:separator algorithms} that {\sf MakeDPDec}$(G)$ returns a DAG-decomposition of $G$, whose width is $O(\log^{3/2}n \cdot dgw(G))$. This gives us a theorem analogous to \thref{thm:dpd-algo} for DAG-width.

Given a directed path decomposition of a digraph $G$, one can construct a Kelly path decomposition of $G$ of the same width using the construction described in the proof of \thref{thm:kpw-equals-dpw}. Note that a Kelly path decomposition is also a Kelly-decomposition. Using \remref{rem:separator algorithms}, we conclude that the output of {\sf MakeDPDec}$(G)$ can be transformed into a Kelly-decomposition of $G$ (in polynomial time), whose width is $O(\log^{3/2}n \cdot kw(G))$. This gives us a theorem analogous to \thref{thm:dpd-algo} for Kelly-width.

\subsection{Approximating directed treewidth}\label{sec:approx-dtw}

\indent In this section, we present an $O(\sqrt{\log n})$-approximation algorithm for directed treewidth. Our algorithm uses ideas from the proof of Theorem 3.3 of Johnson et al. \cite{dtw-definition} which basically states that a digraph with a {\it haven} of large order has a large directed treewidth. Havens correspond to certain winning strategies in a variant of the cops-and-robber game introduced in \cite{tw-visible-dynamic}.\\
\indent Let $G$ be a digraph, and $W, Y \subseteq V(G)$ such that $W$ is $Y$-normal. Let $S$ be a directed vertex separator of $Y$ in $G$ such that $S \cap W \ne \emptyset$. Let $C_1, ..., C_q$ be the strongly connected components of $G \setminus S$. For each $C_i$, consider the strongly connected components of $C_i \setminus Y$. It follows from \remref{rem:normal} that the vertex set of each such strongly connected component is either entirely contained in $W$, or otherwise disjoint from $W$.

Now, let $G_1, ..., G_p$ be all the digraphs such that each $G_j$ is a strongly connected component of $C_i \setminus Y$ for some $1 \le i \le q$, and $V(G_j) \subseteq W$. We say that $W_1:=V(G_1), ..., W_p:=V(G_p)$ is the {\it refinement of $W$ with respect to $Y$ and $S$}. For such a refinement, we define an associated many-to-one function $parent$ from $\{G_1, ..., G_p\}$ to $\{C_1, ..., C_q\}$ as follows: $parent(G_j) := C_i$ whenever $G_j$ is a strongly connected component of $C_i \setminus Y$. Let $\mathcal{D}_1, \mathcal{D}_2, ..., \mathcal{D}_p$ be such that for each $1\le i \le p$, $\mathcal{D}_i$ is an arboreal decomposition of $G$ with respect to $W_i$. We now describe how one can obtain an arboreal decomposition of $G$ with respect to $W$ by gluing the decompositions $\mathcal{D}_1, \mathcal{D}_2, ..., \mathcal{D}_p$.
\begin{sub}[Glue]\label{sub:glue}\footnote{This construction is based on the proof of Theorem 3.3 of Johnson et al. \cite{dtw-definition}.}
We write $\mathcal{D}:={\rm{\sf Glue}}(\mathcal{D}_1, ..., \mathcal{D}_p ; W,Y,S)$ when the decomposition $\mathcal{D}=(T,\mathcal{W},\mathcal{X})$ is constructed as follows: $T$ is obtained by taking the union of $T_1, ..., T_p$, and adding a new (root) node $r_0$ and an arc $e_i$ from $r_0$ to the unique root of $T_i$ for each $1 \le i \le p$. The arc bag $X_{e_i}$ is set to $Y_i:=S \cup (Y \cap V(parent(G_i)))$. The node bag $W_{r_0}$ is set to $S \cap W$. For every other node (arc), the node bag (arc bag) is unchanged.
\end{sub}
\begin{clm}\label{clm:glue}
$\mathcal{D}:={\rm{\sf Glue}}(\mathcal{D}_1, ..., \mathcal{D}_p ; W,Y,S)$ is an arboreal decomposition of $G$ with respect to $W$.
\end{clm}
\begin{proof}
First we show that the node bag of $\mathcal{D}$ form a partition of $W$. Since $\mathcal{D}_j$ is an arboreal decomposition of $G$ with respect to $W_j$ for each $1 \le j \le p$, it suffices to show that $\{W_{r_0}, W_1, W_2, ..., W_p\}$ is a partition of $W$, and that $W_j$ is $Y_j$-normal for each $1 \le j \le p$.\\
\indent Indeed, note that each $W_j$ is the vertex set of a strongly connected component $G_j$ of $G \setminus (S \cup Y)$. It follows that $W_j$ is a subset of $V(H)$, where $H$ is some strongly connected component of $G \setminus Y$. Since $W$ is $Y$-normal, it follows from \remref{rem:normal} that either $V(H) \subseteq W$, or otherwise $V(H) \cap W = \emptyset$. However, by choice of $G_1, ..., G_p$, we know that $V(G_j) \subseteq W$. Note that all other strongly connected components of $G \setminus (S \cup Y)$ are disjoint from $W$. This implies that $W_1 \cup ... \cup W_p = (W-S)-Y = W-(S \cap W)$, where the final equality holds since $W$ and $Y$ are disjoint. Note that $W_{r_0}$ is defined as $S \cap W$. It follows that the union of $W_{r_0}, W_1, ..., W_p$ is $W$. By definition, $W_i$'s are all non-empty and pairwise disjoint, and they are all disjoint from $W_{r_0}$. Also, $W_{r_0}$ is non-empty by definition of $S$.\\
\indent Let $C_i:=parent(G_j)$. Note that $G_j$ is a strongly connected component of $G \setminus (S \cup (Y \cap V(C_i)))$. It follows from the definitions of $W_j$ and $Y_j$, and \remref{rem:normal} that $W_j$ is $Y_j$-normal for each $1 \le j \le p$. This completes the proof.
\end{proof}

The {\sf Glue} subroutine suggests a natural divide-and-conquer algorithm. We now formally describe the recursion subroutine {\sf MakeArbDec}, which receives three arguments: the digraph $G$, and $W, Y \subseteq V(G)$ such that $W$ is $Y$-normal, and returns an arboreal decomposition of $G$ with respect to $W$.\\

\noindent {\sf MakeArbDec}$(G, W, Y)$
\begin{enumerate}
\item {\it Termination Step:} If $|W| \le |Y|$, return the trivial arboreal decomposition of $G$ with respect to $W$.
\item {\it Divide Step:} Let $S':={\rm{\sf FindSep}}(G, Y, \frac{7}{8})$. If $S' \cap W=\emptyset$, let $S:=S' \cup \{v\}$ where $v$ is an arbitrary element of $W$. Otherwise, $S:=S'$.

Let $W_1:=V(G_1), ..., W_p:=V(G_p)$ be the refinement\footnote{Refer to the paragraph preceding \subroutineref{sub:glue}.} of $W$ with respect to $Y$ and $S$. Let $Y_j:=S \cup (Y \cap V(parent(G_j)))$ for each $1 \le j \le p$. Recursively compute $\mathcal{D}_j:={\rm{\sf MakeArbDec}}(G, W_j, Y_j)$ for each $1 \le j \le p$.
\item {\it Combine Step:} Let $\mathcal{D}:={\rm{\sf Glue}}(\mathcal{D}_1, ..., \mathcal{D}_p; W,Y,S)$. Return $\mathcal{D}$.
\end{enumerate}

Given a digraph $G$, our algorithm is just {\sf MakeArbDec}$(G,V(G),\emptyset)$.
\begin{clm}\label{claim:invariant}
At each invocation of {\sf MakeArbDec}, it holds that $|Y| \le 16 \beta \sqrt{\log{n}} \cdot dtw(G)$.
\end{clm}
\begin{proof}
We prove this by induction on the level of recursion. It is trivially true at the first invocation since $Y = \emptyset$. Now, consider a certain invocation whose input parameters are $G$ and $W,Y \subseteq V(G)$. Assume that $|Y| \le 16 \beta\sqrt{\log{n}} \cdot dtw(G)$ holds. Let $S'$ and $S$ be as defined in the divide step. It follows from \cororef{cor:separator algorithms} that $|S'| \le \beta \sqrt{\log{n}} \cdot dtw(G)$. Hence, it holds that $|S| \le 2\beta \sqrt{\log{n}} \cdot dtw(G)$. Let $W_1:=V(G_1), W_2:=V(G_2),...,W_p:=V(G_p)$ be the refinement of $W$ with respect to $Y$ and $S$. Now, it suffices to show that $|S \cup (Y \cap V(parent(G_j)))| \le 16\beta\sqrt{\log{n}}\cdot dtw(G)$ for each $1 \le j \le p$. Note that $V(parent(G_j))$ is the vertex set of a strongly connected component of $G \setminus S$, and $S$ is a $\frac{7}{8}$-balanced directed vertex separator of $Y$ in $G$. It follows that $|Y \cap V(parent(G_j))| \le \frac{7}{8}|Y| \le 14 \beta\sqrt{\log{n}}\cdot dtw(G)$. Hence, $|S \cup (Y \cap V(parent(G_j)))| \le |S| + |Y \cap V(parent(G_j))| \le 16\beta\sqrt{\log{n}}\cdot dtw(G)$.
\end{proof}
Now we are in a position to prove the correctness and approximation guarantee of our algorithm:
\begin{lem}
Given a digraph $G$, {\sf MakeArbDec}$(G,V(G),\emptyset)$ returns an arboreal decomposition of $G$, whose width is $O(\sqrt{\log{n}}\cdot dtw(G))$.
\end{lem}
\begin{proof}
Consider an invocation of {\sf MakeArbDec} which satisfies the condition of the termination step. Let the input parameters be $G$ and $W,Y \subseteq V(G)$. It follows from \remref{rem:trivial} that the invocation returns an arboreal decomposition of $G$ with respect to the set $W$. Let $\mathcal{D}=(T,\mathcal{W},\mathcal{X})$ be the decomposition returned by {\sf MakeArbDec}$(G,V(G),\emptyset)$. Repeatedly applying \clref{clm:glue}, we conclude that $\mathcal{D}$ is an arboreal decomposition of $G$. It remains to show that for each node $i$ of $T$, the size of $W_i \cup X_{\sim i}$ is $O(\sqrt{\log{n}}\cdot dtw(G))$. Note that $i$ corresponds to a unique invocation of {\sf MakeArbDec}. Let $G$ and $W,Y \subseteq V(G)$ be the input parameters for this invocation. It follows from \clref{claim:invariant} that $|Y| \le 16\beta\sqrt{\log{n}}\cdot dtw(G)$. We divide the rest of the proof into two cases depending on whether $i$ is a leaf node or not:\\
{\it Case $1:$} Suppose $i$ is a leaf node. The termination condition holds true for this invocation, i.e. $|W| \le |Y|$. In this case, $|W_i \cup X_{\sim i}| = |W_i| + |X_e|$ where $X_e$ is the unique incoming arc (if any) at node $i$. Observe that $W_i$ and $X_e$ are disjoint since $W_i$ is $X_e$-normal. Since $W_i=W$ and $X_e=Y$, we conclude that $|W_i \cup X_{\sim i}| \le 32\beta\sqrt{\log{n}}\cdot dtw(G)$.\\
{\it Case $2:$} Suppose $i$ is not a leaf node. Let $S'$ and $S$ be as defined in the divide step of this invocation. \cororef{cor:separator algorithms} implies that $|S'| \le \beta \sqrt{\log{n}} \cdot dtw(G)$. Hence, we have $|S| \le 2\beta \sqrt{\log{n}} \cdot dtw(G)$. From \subroutineref{sub:glue}, it follows that $W_i = S \cap W$ and $X_{\sim i} \subseteq S \cup Y$. We conclude that $|W_i \cup X_{\sim i}| \le |S| + |Y| \le 18\beta\sqrt{\log{n}}\cdot dtw(G)$.
\end{proof}

\begin{theorem}
There exists a polynomial time approximation algorithm that, given a digraph $G$, computes an arboreal decomposition of $G$, whose width is $O(\sqrt{\log{n}}\cdot dtw(G))$.
\end{theorem}

\section{Discussion}\label{sec:discussion}

We have presented approximation algorithms for several digraph width parameters. To the best of our knowledge, these are the first (non-trivial) approximation algorithms for each of these parameters. Austrin et al. \cite{tw-inapproximability} have shown that assuming the Small Set Expansion conjecture, treewidth and pathwidth are both hard to approximate within any constant factor. Since all the considered width measures are generalizations of either treewidth or pathwidth, it follows that a similar inapproximability result holds for each of these. The natural question that arises is whether one can design algorithms with better approximation guarantees for these width measures, or otherwise establish stronger inapproximability results. In particular, our approximation algorithms for DAG-width and Kelly-width are implied by our construction of approximate directed path decompositions. This suggests that it might be possible to improve upon our results by computing DAG-decompositions and Kelly-decompositions directly.\\
\indent A limitation of our algorithms is that the approximation guarantees depend on the size of the input digraph. For the considered digraph width measures, most problems (such as Hamiltonian path) which are solvable in polynomial time on digraphs of bounded width are known to admit only {\sf XP} algorithms (as opposed to {\sf FPT}). Some of these problems are also known to be {\sf W[2]}-hard \cite{whardness1}, \cite{whardness2}. In view of these results, it would be desirable to have algorithms whose approximation guarantees depend only on the considered width parameter. As discussed in \secref{sec:intro}, such algorithms exist for (undirected) treewidth {\cite{eyalamir}, \cite{fhl-tw-approx}. However, designing such algorithms for the considered parameters seems to require more sophisticated techniques and we leave it as an open problem.
\subparagraph*{Acknowledgements} We gratefully acknowledge helpful discussions with Robin Thomas.  He motivated us to look at Theorem 3.3 from \cite{dtw-definition}. We thank James Lee and Yury Makarychev for answering our questions about directed separators. Finally, the second author would like to extend special thanks to Joseph Cheriyan and Konstantinos Georgiou whose help was indispensable in improving the presentation.

\bibliographystyle{alpha}
\bibliography{bib-approx-directed-width}

\appendix

\section{Cops-and-robber games}\label{ap:cops-and-robber}
In this Appendix, we give a proof of \propref{prop:kpw-kw-gap} which says that the gap between Kelly pathwidth and Kelly-width can be arbitrarily large. In order to do this, we first describe informally two variants of the cops-and-robber game on a digraph. The first of these is the {\it visible and dynamic} variant, and the second is the {\it invisible and inert} variant. Unless otherwise stated, the description that follows applies to both of these variants.\\
\indent Let $G$ be the digraph on which the game is being played. There are two players, a {\it cop player} and a {\it robber player}. The cop player has $k$ tokens called {\it cops}, and the robber player has a single token called the {\it robber}. We refer to the robber player and the robber interchangeably. The players take turns alternately during which they place their tokens on the vertices of $G$. The cop player's objective is to capture the robber by having a cop occupy the same vertex as the robber. Initially, there are no cops on $G$, and the robber can be occupying any vertex. At each turn of the cop player, he moves an arbitrary subset of cops to any vertices. It is convenient to think of the cops as using helicopters. Let the vertices occupied by the cops be $X$. In case the robber is not captured, the cop player announces where the cops will go in the next turn, say $X'$. At this point, all the cops which are occupying vertices other than $X \cap X'$ remove themselves completely from $G$. One may think of these cops as being in their helicopters. It is during this transition that the robber takes his turn. He evades capture (if possible) by running from his current position along any directed path $P$ which is free of cops, i.e. $P$ does not use a vertex of $X \cap X'$.\\
\indent In the {\it dynamic and visible} variant, the robber is visible to the cops and he may choose to move whenever it is his turn to do so. In the {\it inert and invisible} variant, the robber is invisible to the cops and he is allowed to move only if a cop is about to land on his current position. Note that being invisible is of advantage to the robber but being inert is of disadvantage.\\
\indent Each of these games has a corresponding {\it monotone} version. The game is called {\it cop-monotone} if the cop player is not allowed to occupy a previously vacated vertex. The game is called {\it robber-monotone} if the set of vertices reachable by the robber at each move is not allowed to expand. The game is called {\it monotone} if it is both cop-monotone and robber-monotone. It turns out that the DAG-width of $G$ is $k$ if and only if $k$ cops can capture the robber in the monotone version of the visible and dynamic cops-and-robber game on $G$. Similarly, the Kelly-width of $G$ is $k$ if and only if $k$ cops can capture the robber in the monotone version of the invisible and inert cops-and-robber game on $G$. We refer the reader to \cite{dagw-journal} and \cite{kellywidth-definition} for a detailed treatment of these versions of the game and proofs of these equivalences respectively.\\
\indent The above mentioned games can also be played on an undirected graph $H$. Let $\overline{H}$ denote the directed graph obtained by replacing each edge $\{u,v\}$ of $H$ by two arcs $(u,v)$ and $(v,u)$. Now, one can play the same game on $\overline{H}$. It turns out that the treewidth of $H$ is $k-1$ if and only if $k$ cops can capture the robber in the monotone version of the visible and dynamic cops-and-robber game on $H$ if and only if $k$ cops can capture the robber in the monotone version of the invisible and inert cops-and-robber game on $H$. The equivalence of treewidth and the visible and dynamic variant was shown by Seymour and Thomas \cite{tw-visible-dynamic}. The equivalence of treewidth and the invisible and inert variant was shown by Dendris et al. \cite{tw-invisible-inert}. This leads us to the next proposition.
\begin{prop}\label{prop:equivalences}
For an undirected graph $H$, the following are equivalent:
\begin{enumerate}
\item $H$ has treewidth $k-1$.
\item $\overline{H}$ has DAG-width $k$.
\item $\overline{H}$ has Kelly-width $k$.
\end{enumerate}
\end{prop}

Now, we can proceed to prove \propref{prop:kpw-kw-gap}.

\noindent {\bf Proof of \propref{prop:kpw-kw-gap}:}\\
\indent Let $T_i$ be the (undirected) complete ternary tree of height $i \ge 2$. It can be easily checked that trees have treewidth $1$, and thus $T_i$ has treewidth $1$. \propref{prop:equivalences} implies that $\overline{T_i}$ has Kelly-width $2$. It is known from \cite{searching-pebbling} that $T_i$ has pathwidth exactly $i$, and it is straightforward to show that $\overline{T_i}$ must therefore have directed pathwidth exactly $i$ as well. Using \thref{thm:kpw-equals-dpw}, $\overline{T_i}$ has Kelly pathwidth exactly $i$. Hence, the family $\{T_i:i \ge 2\}$ proves the claim.\hfill$\qed$


\end{document}